\documentclass[]{article}

\usepackage{fullpage}
\usepackage{amsthm}
\usepackage{orcidlink}
\usepackage{amsmath}
\usepackage{amssymb}

\newtheorem{proposition}{Proposition}
\newtheorem{theorem}{Theorem}
\newtheorem{lemma}{Lemma}
\newtheorem{corollary}{Corollary}

\theoremstyle{definition}
\newtheorem{definition}{Definition}
\newtheorem{example}{Example}
\theoremstyle{remark}
\newtheorem{remark}{Remark}
\theoremstyle{plain}
\newtheorem*{keyword}{Keywords}
\newcommand{\sep}{\quad\textbullet\quad}

\begin{document}

\newcommand{\myfunding}{The project has received funding from the European Research Council (ERC) under the European Union’s Horizon 2020 research and innovation programme under grant agreement No~787367 (PaVeS)}
\newcommand{\myemail}{raskin@mccme.ru, raskin@in.tum.de}
\newcommand{\myaddress}{Department of Computer Science, TU Munich 
%\\ Boltzmannstraße 3, D-85748 Garching bei Munich, Germany
}

\title{Population protocols with unreliable communication\footnote{\myfunding}
}

\author{Mikhail Raskin\,\orcidlink{0000-0002-6660-5673}\\ \myemail \\ \myaddress}
%\tnotetext[funding]{\myfunding}
%\ead{\myemail}
%\address{\myaddress}

\maketitle

%\textit{\small This work is not eligible for student paper awards}

\begin{abstract}
Population protocols are a model of distributed computation intended for the study of networks of independent computing agents with dynamic communication structure.  Each agent has a finite number of states, and communication occurs nondeterministically, allowing the involved agents to change their states based on each other's states.

In the present paper we study unreliable models based on population protocols and their variations from the point of view of expressive power.  We model the effects of message loss.  We show that for a general definition of protocols with unreliable communication with constant-storage agents such protocols can only compute predicates computable by immediate observation (IO) population protocols (sometimes also called one-way protocols).  Immediate observation population protocols are inherently tolerant to unreliable communication and keep their expressive power under a wide range of fairness conditions.  We also prove that a large class of message-based models that are generally more expressive than IO becomes strictly less expressive than IO in the unreliable case.

\end{abstract}

\begin{keyword}
population protocols \sep message loss \sep expressive power
\end{keyword}

%\begin{acks}
%{I thank Javier Esparza
%for useful discussions and
%the feedback on the drafts of the present article.
%I thank Chana Weil-Kennedy for useful discussions.
%I thank the anonymous reviewers for their valuable feedback on presentation.
%
%{The project has received funding from the European Research Council (ERC) under the European Union’s Horizon 2020 research and innovation programme under grant agreement No~787367 (PaVeS)}
%        }
%\end{acks}

\section{Introduction}

Population protocols have been
introduced in \cite{conf/podc/AngluinADFP04,journals/dc/AngluinADFP06}
as a restricted yet useful subclass of general distributed protocols.
In population protocols each agent has a constant amount of local storage,
and during the protocol execution pairs of agents
are selected and permitted to interact.
The selection of pairs is assumed to be done by an adversary
bound by a fairness condition.
The fairness condition ensures
that the adversary cannot trivially stall the protocol.
A typical fairness condition requires
that every configuration that stays reachable
during an infinite execution
is reached infinitely many times.

Population protocols have been studied
from various points of view,
such as expressive power \cite{journals/dc/AngluinAER07},
verification complexity \cite{conf/concur/EsparzaGMW18},
time to convergence \cite{conf/wdag/AngluinAE06,conf/wdag/DotyS15},
privacy \cite{conf/opodis/Delporte-GalletFGR07},
impact of different interaction scheduling \cite{conf/opodis/ChatzigiannakisDFMS09} etc.
Multiple related models have been introduced.
Some of them change or restrict
the communication structure: this is the case for
immediate, delayed, and queued transmission and observation \cite{journals/dc/AngluinAER07},
as well as for broadcast protocols \cite{conf/lics/EmersonN98}.
Some explore the implications of adding limited amounts of storage
(below the usual linear or polynomial storage permitted in
traditional distributed protocols):
this is the case for
community protocols \cite{Guerraoui07abstract}
(which allow an agent to recognise a constant number
of other agents),
PALOMA \cite{journals/corr/abs-1004-3395}
(permitting logarithmic amount of local storage),
mediated population protocols \cite{journals/tcs/MichailCS11}
(giving some constant amount of common
storage to every pair of agents),
and others.

The original target application
of population protocols and related models
is modelling networks of restricted sensors,
starting from
the original paper \cite{conf/podc/AngluinADFP04} on population protocols.
On the other hand, verifying distributed algorithms benefits
from translating the algorithms in question or their parts
into a restricted setting, as most problems are undecidable
in the unrestricted case.
Both applications motivate study of fault tolerance.
Some papers on population protocols
and related models
\cite{conf/dcoss/Delporte-GalletFGR06,Guerraoui07abstract,Angluin08asimple,Guerraoui_namestrump}
consider questions of fault tolerance,
but in the context of expressive power the fault is typically
expected to be either a total agent failure
or a Byzantine failure.
There are some exceptions such as a study
of fine-grained notions of unreliability
\cite{diluna2016population,diluna2016power}
in the context of step-by-step simulation of population protocols
by distributed systems with binary interactions.
However, these studies answer a completely different set of questions,
as they are concerned with simulating a protocol as a process
as opposed to designing a protocol to achieve a given result
no matter in what way.

In a practical context, many distributed algorithms
pay attention to a specific kind of failure: message loss.
While the eventual convergence approach
typical in study of population protocols
escapes the question of availability \emph{during}
a temporary network partition
(the problem studied, for example, in \cite{friedman1996trading}),
the onset of a network partition may include
message loss in the middle of an interaction.
In such a situation the participants do
not always agree
whether the interaction has succeeded or failed.
In terms of population protocols,
one of the agents assumes that an interaction
has happened and updates the local state,
while a counterparty thinks the interaction has failed
and keeps the old state.

In the present paper we study the expressive power of 
a very wide class of
models with interacting constant-storage agents when
unreliability of communication is introduced.
This unreliability corresponds to the loss of atomicity of interactions
due to message loss.
Indeed, in the distributed systems ensuring that both sides agree 
on whether the interaction has taken place is often the costliest part;
a special case of it is ``exactly-once'' message arrival,
known to be much
more complex to ensure than ``at most-once''.
We model such loss of atomicity by allowing some agents to update their
state based on an interaction, while other agents keep their original 
state because they assume the interaction has failed.
For a bit more generality, corresponding, for example,
to request-response interactions 
with the response being impossible if the request is lost,
we allow to require that some agents can only update their state 
if the others do.

We consider the expressive power in the context
of computing predicates by protocols with eventual convergence
of individual opinions.
We show that under very general conditions
the expressive power of  protocols with unreliable communication
coincides with the expressive power
of immediate observation population protocols.
Immediate observation population protocols, 
modelling interactions where an agent can observe
the state of another one
without the observee noticing,
provide
a model that inherently tolerates unreliability
and is considered a relatively weak model
in the fully reliable case.
This model also has other nice properties,
such as relatively low complexity ($\mathbf{PSPACE}$-complete)
of verification tasks \cite{esparza2019parameterized}.
Our results hold under any definition of fairness
        satisfying two general assumptions
        (see Definition~\ref{definition:execution}),
including all the usually used versions of fairness.

We prove it by observing a general structural property
shared by all  protocols with unreliable communication.
Informally speaking, protocols with unreliable communication have some
special fair executions which can be extended by adding
an additional agent with the same initial and final state
as a chosen existing one.
This property is similar to the copycat arguments
used, for example, for proving the exact expressive power
of immediate observation protocols.
The usual structure of the copycat arguments
includes a proof that we can pick an agent in an execution
and add another agent (copycat) which will repeat all the
state transitions of the chosen one.
In the immediate observation case the corresponding
property is almost self-evident once defined.
A slightly stronger but still straightforward argument is needed
in the case of reconfigurable broadcast networks \cite{BertrandBM19}.
The latter model is equivalent to unreliable broadcast networks;
a sender broadcasts a message and changes the local state,
and an arbitrary set of receivers react to the message (immediately).
However,  unlike all the previous uses of the copycat-like arguments
in the context of population protocols and similar models,
proving the necessary copycat-like property
for a general notion of  protocols with unreliable communication
(sufficient to handle assymetry of message loss where loss for
sender requires loss for receiver)
requires careful analysis using different techniques.

Note that although the natural way to design population protocols
for our setting involves the use of immediate observation population protocols,
we still need to rule out additional opportunities arising
from the fact that eventually a two-agent interaction
with both agents correctly updated will happen.
However, in contrast to self-stabilising protocols 
\cite{DBLP:journals/cacm/Dijkstra74,journals/taas/AngluinAFJ08},
the protocols cannot rely on the message loss being absent
for an arbitrarily long time.

Surprisingly, asynchronous transmission and receipt of messages,
which provides more expressive power
than immediate observation population protocols
in the reliable setting, turns out to have strictly less
expressive power in the unreliable setting.
Note that message reordering is allowed already in the reliable setting,
while unreliability is essentially a generalisation of message loss.
One could say that an unbounded delay in message delivery 
becomes a liability instead of an asset once there is message loss.

The rest of the present paper is organised as follows.
First, in Section~\ref{sec:basic-def} we define a general protocol framework generalising
many previously studied approaches.
Then in Section~\ref{sec:known-expressive-power} we summarise the results from the literature on
the expressive power of various models covered by this framework.
Afterwards in Section~\ref{sec:models} we formally define our general notion of
a protocol with unreliable communication.
%Then in Section~\ref{sec:shadow-agent} we formulate and prove the common limitation of all the  protocols with unreliable communication.
%This allows us to conclude the proof of the main result in the Section~\ref{sec:main-proof}.
Then in Section~\ref{sec:main-proof} we formalise the common limitation of all the  protocols with unreliable communication,
and provide the proof sketches of this restriction and the main result.
Afterwards in Section~\ref{sec:message-based} we show that fully asynchronous (message-based) models
become strictly less powerful
than immediate observation in the unreliable setting.
The paper ends with a brief conclusion and some possible future directions.

%Due to the space constraints the detailed proofs are provided in the full version \cite{raskin2021population}.
%For agreement with the published version the detailed proofs are provided in the appendix.

\subsection{Main results (preview)}

The precise statements of our results require the detailed definitions introduced
later. However, we can roughly summarise them as follows.

First, we characterise the expressive power of all fixed-memory protocols given unreliable comunication.

\begin{proposition}
Adding unreliability of communication
to population protocols restricts the predicates they can express
to boolean combinations of comparisons of arguments with constants.
\end{proposition}

This is the same expressive power as the immediate observation protocols.

Next we show that unreliability changes the expressive power non-monotonically for some natural classes.

\begin{proposition}
Queued transmission protocols with unreliable communication are strictly less expressive
        than immediate observation population protocols (with or without unreliable communication).
\end{proposition}

Note that without unreliability 
queued transmission protocols are strictly more expressive than immediate observation population protocols.

\section{Basic definitions}
\label{sec:basic-def}

\subsection{Protocols}

We consider various models of distributed computation where
the number of agents is constant during protocol execution,
each agent has a constant amount of local storage,
and agents cannot distinguish each other except via the states.
We provide a general framework for describing such protocols.
Note that we omit some very natural restrictions
(such as decidability of correctness of a finite execution)
because they are irrelevant for the problems we study.
We allow agents
to be distinguished and tracked individually
for the purposes of analysis,
even though they cannot identify each other
during the execution of the protocol.

We will use the following problem to illustrate our definitions:
the agents have states $q_0$ and $q_1$ corresponding to input
symbols $0$ and $1$ and aim to find out if all the agents
have the same input.
They have an additional state $q_\bot$
to represent the observation that both input symbols were present.
We will define four protocols for this problem
using different communication primitives.
\begin{itemize}
        \item Two agents interact and both switch to $q_\bot$
                unless they are in the same state
                (population protocol interaction).
        \item An agent observes another agent and switches
                to $q_\bot$ if they are in different states
                (immediate observation).
        \item An agent can send a message with its state,
                $q_0$, $q_1$ or $q_\bot$.
                An agent in a state $q_0$ or $q_1$
                can receive a message (any of the pending
                messages, regardless of order);
                the agent switches
                to $q_\bot$ if the message contains a state
                different from its own
                (queued transmission).
        \item An agent broadcasts its state without
                changing it;
                each other agent receives the broadcast
                simultaneouly and
                switches to $q_\bot$ if its state is
                different from the broadcast state
                (broadcast protocol interaction).
\end{itemize}

\begin{definition}
A \emph{protocol} is specified by a tuple
        $(Q,M,\Sigma,I,o,\mathrm{Tr},\Phi)$,
with components being
        a finite nonempty set $Q$ of (individual agent) \emph{states},
        a finite (possibly empty) set $M$ of \emph{messages},
        a finite nonempty \emph{input alphabet} $\Sigma$,
        an \emph{input mapping function} $I: \Sigma\to Q$,
        an \emph{individual output function} $o: Q\to \{true,false\}$,
        a \emph{transition} relation $\mathrm{Tr}$ (which is described in more details below),
        and a \emph{fairness condition} $\Phi$ on executions.
\end{definition}

The protocol defines evolution of populations of agents
(possibly with some message packets being present).

\begin{definition}
        A \emph{population} is a pair of sets:
        $A$ of \emph{agents} and $P$ of \emph{packets}.
        A \emph{configuration} $C$ is a population together
        with two functions, $C_A: A\to Q$ provides agent states,
        and $C_P: P\to M$ provides packet contents.
\end{definition}

        Note that if $M$ is empty, then $P$ must also be empty.
        As the set of agents is the domain of the function $C_A$,
        we use the notation $\mathrm{Dom}(C_A)$ for it.
        The same goes for the set of packets $\mathrm{Dom}(C_P)$.
        Without loss of generality $\mathrm{Dom}(C_P)$
        is a subset of a fixed countable set of possible packets.

        The message packets are only used for asynchronous communication;
        instant interaction between agents
        (such as in the classical rendezvous-based population protocols
        or in broadcast protocols)
        does not require describing the details of communication 
        in the configurations.

\begin{example}
        The four example protocols have
        the same set of states $Q=\{q_0,q_1,q_\bot\}$.
        The first two protocols have the empty set of messages,
        and the last two have the set of messages $M=\{m_0,m_1,m_\bot\}$.
The example protocols all have the same
        input alphabet $\Sigma=\{0,1\}$,
        input mapping $I: i\mapsto q_i$,
        and output mapping
        $o: q_0\mapsto \mathrm{true}, q_1\mapsto \mathrm{true}, q_\bot\mapsto \mathrm{false}$.
\end{example}

The definition of the transition relation uses the following notation.

\begin{definition}
        For a function $f$ and $x\notin \mathrm{Dom}(f)$ let $f\cup\{x\mapsto y\}$
        denote the function $g$ defined on $\mathrm{Dom}(f)\cup \{x\}$ such that
        $g\mid_{\mathrm{Dom}(f)}=f$ and $g(x)=y$.
        For $u\in \mathrm{Dom}(f)$ let $f[u\mapsto v]$ denote
        the function $h$ defined on $\mathrm{Dom}(f)$ such that
        $h\mid_{\mathrm{Dom}(f)\setminus\{u\}}=f\mid_{Dom(f)\setminus\{u\}}$ and $h(u)=v$.
        For symmetry, if $w=f(u)$
        let $f\setminus\{u\mapsto w\}$ denote
        restriction $f\mid_{\mathrm{Dom}(f)\setminus\{u\}}$.

        Use of this notation implies an assertion of correctness,
        i.e. $x\notin \mathrm{Dom}(f)$, $u\in \mathrm{Dom}(f)$, and $w=f(u)$.
        We use the same notation with a configuration $C$ instead of a function
        if it is clear from context whether $C_A$ or $C_P$ is modified.
\end{definition}

Now we can describe the transition relation that tells us which
configurations can be obtained from a given one
via a single interaction.
In order to cover broadcast protocols
we define the transition relation as a relation on configurations.
The restrictions on the transition relation
ensure that the protocol behaves like
a distributed system
with arbitrarily large number of anonymous agents.

\begin{definition} The  \emph{transition} relation of a protocol is a set of triples $(C,A^\odot,C')$, called  \emph{transitions}, where $C$ and $C'$ are configurations and  $A^\odot \subset \mathrm{Dom}(C_A)$ is the set of \emph{active agents} (of the transition); agents in $\mathrm{Dom}( C_A)\setminus A^\odot$, are called \emph{passive}.
        We write $C\xrightarrow{A^\odot} C'$ for $(C, A^\odot, C')\in \mathrm{Tr}$,
        and let $C\to C'$ denote the projection
        of $\mathrm{Tr}$: $C\to{}C'\Leftrightarrow{}\exists{A^\odot}:C\xrightarrow{A^\odot}C'$.
The transition relation must satisfy the following conditions for every  transition $C\xrightarrow{A^\odot} C'$:
\begin{itemize}
        \item \textbf{Agent conservation}. $\mathrm{Dom}( C_A) = \mathrm{Dom}( C'_A)$.
        %\item \textbf{Packet immutability}. For every $p \in \mathrm{Dom}( C_P)\cap\mathrm{Dom}( C'_P)$: $C_P(p)=C'_P(p)$.
        \item \textbf{Agent and packet anonymity}. If $h_A$ and $h_P$ are bijections
                such that $D_A = C_A \circ h_A$, $D'_A = C'_A \circ h_A$,
                $D_P= C_P \circ h_P $, and $D'_P = C'_P \circ h_P $,
                then $D\xrightarrow{h^{-1}(A^\odot)}D'$.
        \item \textbf{Possibility to ignore extra packets}. For every $p \notin \mathrm{Dom}( C_P) \cup\mathrm{Dom}( C'_P)$ and $m \in M$:
               $C \cup \{p \mapsto m\} \xrightarrow{A^\odot} C' \cup \{p \mapsto m\}$.
       \item \textbf{Possibility to add passive agents}. For every agent $a \notin \mathrm{Dom}( C_A)$ and  $q \in Q$ there
                 exists $q' \in Q$ such that: $C \cup \{a \mapsto q\} \xrightarrow{A^\odot} C' \cup \{a \mapsto q'\}$.
         %\item \textbf{Irrelevance of state of passive agents}. For every passive agent $a\in \mathrm{Dom}(C_A)\setminus A^\odot$ and $q \in Q$ there exists $q' \in Q$ such that: $C [a \mapsto q] \xrightarrow{A^\odot} C [a \mapsto q']$
\end{itemize}
\end{definition}

Informally speaking, the active agents are the agents that transmit something
during the interaction.
The passive agents can still observe other agents and change their state.
The choice of active agents
is used for the definition of  protocols with unreliable communication, as a failure
to transmit precludes success of reception.
The formal interpretation will be provided in Definition~\ref{def:unreliable-protocol}.

Many models studied in the literature have the transition relation defined
using pairwise interaction. In these models the transitions are always
changing the states of two agents based on their previous states.
When discussing such protocols, we will use the notation
$(p,q)\to(p',q')$
for a transition where agents in the states $p$ and $q$ switch
to states $p'$ and $q'$, correspondingly.

\begin{example}
The four example protocols have the following transition relations.
\begin{itemize}
        \item In the first protocol for a configuration $C$
                and two agents $a,a'\in \mathrm{Dom}(C_A)$ such that
                $C_A(a)\ne C_A(a')$
                we have
                $C\xrightarrow{\{a,a'\}}
                C[a\mapsto q_\bot][a'\mapsto q_\bot]$
                (in other notation,
                $(C,\{a,a'\},
                C[a\mapsto q_\bot][a'\mapsto q_\bot])\in \mathrm{Tr}$).
        \item In the second protocol for a configuration $C$
                and two agents $a,a'\in \mathrm{Dom}(C_A)$ such that
                $C_A(a)\ne C_A(a')$
                we have
                $C\xrightarrow{\{a\}}
                C[a\mapsto q_\bot]$.
                We can say that $a$  observes $a'$
                in a different state and switches to $q_\bot$.
        \item In the third protocol there are two types of transitions.
                Let a configuration $C$ be fixed.
                For an agent $a\in \mathrm{Dom}( C_A)$,
                $i\in\{0,1,\bot\}$ such that $C_A(a)=q_i$,
                and a new message identity $p\notin \mathrm{Dom}( C_P)$
                we have $C\xrightarrow{\{a\}} C\cup\{p\mapsto m_i\}$
                (sending a message).
                If $C_A(a)=q_i$ for some $i\in\{0,1\}$,
                for each message $p\in \mathrm{Dom}(C_P)$,
                we also have
                $C \xrightarrow{\{a\}}
                C[a\mapsto q']\setminus\{p\mapsto C_P(p)\}$
                where $q'$ is equal to $q_i$ if $C_P(p)=m_i$
                and $q_\bot$ otherwise
                (receiving a message).
        \item In the fourth protocol,
                for a configuration $C$
                and an agent $a\in \mathrm{Dom}( C_A)$
                we can construct $C'$
                by
                replacing $C_A$ with $C'_A$ that maps
                each $a'\in \mathrm{Dom}( C_A)$ to $C_A(a')$ if $C_A(a)=C_A(a')$
                and $q_\bot$ otherwise.
                Then we have
                $C\xrightarrow{\{a\}}C'$.
                We can say that $a$ broadcasts its state
                and all the agents in the other states
                switch to $q_\bot$.
\end{itemize}
\end{example}

\subsection{Definitions of the protocol classes studied in the literature}

Many previously studied models can be defined inside out framework.
Among such models are population protocols, 
immediate transmission population protocols,
immediate observation population protocols,
queued transmission protocols, broadcast protocols.
These general definitions are similar to the definitions
for specific protocols provided as exampled,
and our results do not depend on these definitions.
We provide them as a corroboration of sufficient 
generality of our framework.

First we translate the initial definition of the population protocols \cite{conf/podc/AngluinADFP04}.

\begin{definition}
        A \emph{population protocol}
        is described by an interaction relation $\delta\subseteq Q^2\times Q^2$.
        The set of messages is empty.
        A configuration $C'$ can be obtained from $C$,
        if there are agents $a_1, a_2\in \mathrm{Dom}( C_A)$
        and states $q_1, q_2, q_3, q_4\in Q$
        such that $C_A(a_1)=q_1$, $C_A(a_2)=q_2$,
        $C'=C[a_1\mapsto q_3][a_2\mapsto q_4]$,
        and $((q_1,q_2),(q_3,q_4))\in \delta$.
        The set of active agents $A^\odot$ is $\{a_1,a_2\}$.
\end{definition}

Now we proceed to the variants of the population protocols
appearing in the paper on expressive power of population protocols
and their variants \cite{journals/dc/AngluinAER07}.

\begin{definition}
        An \emph{immediate transmission population protocol}
        is a population protocol such that
        $q_3$ depends only on $q_1$,
        i.e. the following two conditions hold.
        If
        $((q_1,q_2),(q_3,q_4))\in \delta$
        and
        $((q_1,q'_2),(q'_3,q'_4))\in \delta$
        then $q_3=q'_3$.
        If
        $((q_1,q_2),(q_3,q_4))\in \delta$
        then for every $q'_2$ there exists $q'_4$ such that
        $((q_1,q'_2),(q_3,q'_4))\in \delta$.
\end{definition}

\begin{definition}
        An \emph{immediate observation population protocol}
        is an immediate transmission population protocol such that
        every possible interaction
        $((q_1,q_2),(q_3,q_4))\in \delta$
        has $q_1 = q_3$.

        We can consider only the first agent to be active.
\end{definition}

\begin{definition}
        \emph{Queued transmission protocol}
        has a nonempty set $M$ of messages.
        It has two transition relations:
        $\delta_s\subseteq Q\times (Q\times M)$
        describing sending the messages,
        and
        $\delta_r\subseteq (Q\times M)\times Q$
        describing receiving the messages.
        If agent $a$ has state $q=C_A(a)$ and
        $(q,(q',m))\in\delta_s$,
        it can
        send a message $m$ as a fresh packet $p$ and switch to state $q'$:
        $C\xrightarrow{\{a\}}C[a\mapsto q']\cup\{p\mapsto m\}$.
        If agent $a$ has state $q=C_A(a)$,
        packet $p$ contains message $m=C_P(p)$ and
        $((q,m),q'))\in\delta_r$,
        agent $a$ can
        receive the message:
        $C\xrightarrow{\{a\}}C[a\mapsto q']\setminus\{p\mapsto m\}$.

        \emph{Delayed transmission protocol}
        is a queued transmission protocol where every message
        can always be received by every agent,
        i.e. the projection of $\delta_r$
        to $Q\times M$ is the entire $Q\times M$.

        \emph{Delayed observation protocol}
        is a delayed transmission protocol where sending a message
        doesn't change state, i.e. $(q,(q',m))\in\delta_s$
        implies $q=q'$.
\end{definition}

As the last example, we consider broadcast protocols \cite{conf/lics/EmersonN98}.

\begin{definition}
        \emph{Broadcast protocol}
        is defined by two relations:
        $\delta_s\subseteq Q\times Q$ describing a sender transition,
        and $\delta_r\subseteq (Q\times Q)\times Q$.
        To perform a transition from a configuration $C$, we pick
        an agent $a\in{}\mathrm{Dom}(C_A)$ with state $q$
        and change its state to $q'$ such that $(q,q')\in\delta_s$.
        At the same time, we simultaneously update the state of all
        other agents, in such a way that an agent
        in state $q_j$ can switch to any state $q'_j$
        such that $((q_j,q),q'_j)\in\delta_r$.

        We consider the transmitting agent to be the only active one.
\end{definition}

\begin{remark}
In the literature,
the relations $\delta$, $\delta_s$, $\delta_r$ and $\delta_s$
are sometimes required to be partial functions.
As we use relations in the general case,
we use relations here for consistency.
\end{remark}

\subsection{Fair executions}

In this section we define the notion of fairness.
This notion is traditionally used to exclude the most pathological cases
without a complete probabilistic analysis of the model.
For the population protocols fairness has been a part of the definition
since the introduction \cite{conf/podc/AngluinADFP04,journals/dc/AngluinADFP06}.
However, in the general study of distributed computation
there has long been some interest in comparing effects of different
approaches to fairness in execution scheduling \cite{DBLP:journals/dc/AptFK88}.
For example, the distinction between weak fairness and strong fairness
and the conditions where one can be made to model the other has been
studied in \cite{DBLP:journals/tpds/Karaata01}.
The difference between weak and strong scheduling is that
strong fairness executes infinitely often every 
interaction that is enabled infinitely often,
while weak fairness only guarantees anything for continuously enabled 
interactions.
As there are multiple notion of fairness in use, we define their basic
common traits.
Our results hold for all notions of fairness satisfying these basic requirements,
including all the notions of fairness used in the literature,
as well as much stronger and much weaker fairness conditions.

\begin{definition}\label{definition:execution}
        An \emph{execution} is a sequence (finite or infinite) $C_n$
        of configurations such that at each moment $i$
        either nothing changes, i.e. $C_i=C_{i+1}$ or a single
        interaction occurs, i.e. $C_i\to C_{i+1}$.
        A configuration $C'$ is \emph{reachable} from configuration $C$
        if there exists an execution $C_0,\ldots,C_n$
        with $C_0=C$ and $C_n=C'$ (and \emph{unreachable} otherwise).

        A protocol defines a \emph{fairness condition} $\Phi$
        which is a predicate on executions.
        It should satisfy the following properties.
\begin{itemize}
        \item A fairness condition is \emph{eventual}, i.e. every
                finite execution can be continued to an infinite
                fair execution.
        \item A fairness condition ensures \emph{activity}, i.e.
                if an execution contains only configuration $C$ after
                some moment,
                only $C$ itself is reachable from $C$.
\end{itemize}
\end{definition}

\begin{definition}
        The \emph{default fairness condition}
        accepts an execution
        if every configuration either becomes unreachable
        after some moment, or occurs infinitely many times.
\end{definition}

\begin{example}
The example protocols use the default fairness condition.
\end{example}

It is clear that the default fairness condition ensures activity.

\begin{lemma}[adapted from
        \cite{journals/dc/AngluinAER07}]
        \label{lm:defaultEventual}
        Default fairness condition is eventual.
\end{lemma}

\begin{proof}
Consider a configuration after a finite execution.
Then there is a countable set of possible configurations
(note that the set of potential packets is at most countable).
Consider an arbitrary enumeration of configurations that mentions
each configuration infinitely many times.

We repeat the following procedure:
skip unreachable configurations in the enumeration,
then perform the transitions necessary to reach the next reachable one.
If we skip a configuration, it can never become reachable again.
Therefore all the configurations that stay reachable infinitely long
are never skipped and therefore they are reached infinitely many times.
\end{proof}

The fairness condition is sometimes said to be an approximation
of probabilistic behaviour.
In our general model
the default fairness condition
provides executions similar
to random ones
for protocols without messages but not always for protocols
with messages.
The arguments from \cite{journals/corr/abs-1912-06578} with minimal modification
prove this.
The core idea in the case without messages is observing 
we have a finite state space 
reachable from any given configuration; a random walk eventually
gets trapped in some strongly connected component, visiting all of its 
states infinitely many time.
If we do have messages, the message count might behave like a biased random walk;
while consuming all the messages stays possible in principle,
with probability one it only happens a finite number of times.

\subsection{Functions implemented by protocols}

In this section we recall the standard notion of a function evaluated by a protocol.
Here the standard definition generalises trivially.

\begin{definition}
        An \emph{input configuration} is a configuration where
        there are no packets and all agents are in input states,
        i.e. $P=\varnothing$ and $\mathrm{Im}(C_A) \subseteq \mathrm{Im}(I)$
        where $\mathrm{Im}$ denotes the image of a function.
        We extend $I$ to be applicable to multisets of input symbols.
        For every $\overline{x}\in\mathbb{N}^\Sigma$,
        we define $I(\overline{x})$ to be a configuration
        of $|\overline{x}|$ agents with $\sum_{I(\sigma)=q_i}\overline{x}(\sigma)$
        agents in input state $q_i$ (and no packets).

        A configuration $C$ is a \emph{consensus} if
        the individual output function
        yields the same value for the states of all agents,
        i.e. $\forall a,a'\in \mathrm{Dom}(C_A): o(C_A(a))=o(C_A(a'))$.
        This value is the output value for the configuration.
        $C$ is a \emph{stable consensus}
        if all configurations reachable from $C$
        are consensus configurations
        with the same value.

        A protocol \emph{implements}
        a predicate $\varphi: \mathbb{N}^\Sigma\to \{true,false\}$
        if for every $\overline{x}\in\mathbb{N}^\Sigma$
        every fair execution starting from $I(\overline{x})$
        reaches a stable consensus
        with the output value $\varphi(\overline{x})$.
        A protocol is \emph{well-specified} if it implements
        some predicate.
\end{definition}

\begin{example}
        It is easy to see that each of the four example protocols implements
        the predicate $\varphi(\overline{x})\Leftrightarrow(x(0)=0)\vee(x(1)=0)$
        on $\mathbb{N}^2$.
        In other words, the protocol accepts the input configurations
        where one of the two input states has zero agents and rejects
        the configurations where both input states occur.
\end{example}

This framework is general enough to define the models studied in the literature,
such as
population protocols, immediate transmission protocols,
immediate observation population protocols,
delayed transmission protocols,
delayed observation protocols,
queued transmission protocols,
and
broadcast protocols.

\section{Expressive power of population protocols and related models}
\label{sec:known-expressive-power}

In this section we give an overview of previously known results
on expressive power of various models related to population protocols.
We only consider predicates, i.e. functions with
the output values being $\mathrm{true}$ and $\mathrm{false}$ because
the statements of the theorems become more straightforward in that case.

The expressive power of models related to
population protocols is
expressed in terms of
semilinear, $\mathbf{coreMOD}$,
and counting predicates.
Semilinear predicates on tuples of natural numbers can be
expressed using the addition function, remainders modulo constants,
and the order relation,
such as $x+x \geq{} y+3$ or $x\mod 7=3$.
Roughly speaking,
$\mathbf{coreMOD}$ is the class of predicates that become
equivalent to modular equality for inputs with only
large and zero components.
An example could be $(z=1 \wedge x\geq{}y) \vee (x+y \mod 2 = 0)$,
a semilinear predicate which becomes a modular equality
whenever $z=0$ or $z$ is large (i.e. $z\geq{}2$).
Counting predicates are logical combinations
of inequalities including one coordinate and one constant each,
for example, $x\geq{3}$.

\begin{theorem}[see \cite{journals/dc/AngluinAER07} for details]
Population protocols and queued transmission protocols
can implement precisely semilinear predicates.

Immediate transmission population protocols
and delayed transmission protocols can implement
precisely all the semilinear predicates that are also in $\mathbf{coreMOD}$.

Immediate observation population protocols
implement  counting predicates.

Delayed observation protocols implement the counting
predicates where every constant is equal to $1$.
\end{theorem}

\begin{theorem}[see \cite{blondin2019expressive} for details]
Broadcast protocols implement precisely the predicates
computable in nondeterministic linear space.
\end{theorem}

% As an example we can consider the protocols with two input states,
% $q_x$ and $q_y$.
% The predicate $x\geq{}y$ can be implemented by population protocols,
% queued transmission protocols, and broadcast protocols.
% It cannot be implemented by immediate observation protocols or
% immediate transmission protocols, as it requires comparison of
% two parameters which cannot be expressed a a finite combination
% of comparisons of single parameters with constants.
% Immediate transmission protocols can implement $x \mod 7 = 3$,
% unlike immediate observation protocols.
% The predicate $x\geq{}5$ can be implemented by immediate observation
% population protocols but not by delayed observation population protocols.
% Delayed observation population protocols can implement $x\geq{}1$.
% The predicate $y=x^2$ can be implemented by broadcast protocols,
% but not by population protocols, queued transmission protocols or weaker models.

\section{Our models}
\label{sec:models}

\subsection{Proposed models}

We propose a general notion of an unreliable communication version
of a protocol.
Our notion models transient failures,
so the set of agents is preserved.
The intuition we formalise is the idea that
for every possible transition
some agents may fail to update their states
(and keep their corresponding old states).
We also require
that for some passive agent to receive a transmission,
the transmission has to occur
(and active agents who transmit
do not update their state if they fail to transmit,
although a successful transmission can still fail to be received).

\begin{definition}
        \label{def:unreliable-protocol}
        A \emph{protocol with unreliable communication},
        corresponding to a protocol $\mathcal{P}$,
        is a protocol that differs from $\mathcal{P}$ only in the transition relation.
        For every allowed transition $C\xrightarrow{A^\odot} C'$ we also allow
        all the transitions $C\xrightarrow{A^\odot} C''$ where
        $C''$ satisfies the following
        conditions.
\begin{itemize}
        \item \textbf{Population preservation}.
                $\mathrm{Dom}(C''_A)=\mathrm{Dom}(C'_A)$, $\mathrm{Dom}(C''_P)=\mathrm{Dom}(C'_P)$.
        \item \textbf{State preservation}.
                For every agent $a\in \mathrm{Dom}(C''_A)$: $C''_A(a)\in\{C_A(a), C'_A(a)\}$.
        \item \textbf{Message preservation}.
                For every packet $p\in \mathrm{Dom}(C''_P)$: $C''_P(p)=C'_P(p)$.
        \item \textbf{Reliance on active agents}.
                Either for every agent $a\notin A^\odot$ we have $C''_A(a) = C_A(a)$,
                or for every agent $a\in A^\odot$ we have $C''_A(a) = C'_A(a)$.
\end{itemize}
\end{definition}

\begin{example}
\begin{itemize}
        \item
Population protocols with unreliable communication allow an interaction
to update the state of only one of the two agents.
        \item
Immediate transmission population protocols
with unreliable communication
allow the sender to update
the state with no receiving agents.
        \item
Immediate observation population protocols with
unreliable communication do not differ
from ordinary immediate observation population protocols,
because each transition changes the state of only one agent.
Failing to change the state means a no-change transition
which is already allowed anyway.
        \item
Queued transmission protocols with unreliable communication allow messages
to be discarded with no effect.
Note that for  delayed observation protocols unreliable communication doesn't
change much, as sending the messages also has no effect.
\item
 Broadcast protocols with unreliable communication allow a broadcast
to be received by an arbitrary subset of agents.
\end{itemize}
\end{example}

\subsection{The main result}

Our main result is that no class of  protocols with unreliable communication
can be more expressive than immediate observation protocols.

\begin{definition}
        A \emph{cube} is a subset of $\mathbb{N}^k$
        defined by a lower and upper  (possibly infinite) bound for
        each coordinate.
        A \emph{counting set} is a finite union of cubes.

        A \emph{counting predicate} is a membership predicate for some
        counting set.
        Alternatively, we can say it is a predicate that can be computed
        using comparisons of input values with constants and logical operations.

\end{definition}

\begin{theorem}\label{thm:implementable}
The set of predicates that can be implemented by
protocols with unreliable communication is the set of counting predicates.
        All counting predicates can be implemented by (unreliable)
        immediate observation protocols.
\end{theorem}

\section{Proof of the main result}
\label{sec:main-proof}

Our main lemma  is generalises of the copycat lemma
normally applied to specific models such as
immediate observation protocols.
The idea is that for every initial configuration there is
a fair execution that can be extended to a possibly unfair
execution by adding a copy of a chosen agent.
In some special cases,
for example,  broadcast protocols with unreliable communication,
a simple proof can be given by
saying that if the original agent participates
in an interaction,
the copy should do the same just before the original
without anyone ever receiving the broadcasts from the copy.
The copycat arguments 
are usually applied to models where a similar proof suffices.
The situation  is more complex for models like immediate
transmission protocols with unreliable communication. 
As a message cannot be received without being sent,
the receiver cannot update its state if the sender doesn't.
We present an argument applicable in the general case.

\begin{definition}
Let $E$ be an arbitrary execution
of protocol~$P$ with initial configuration $C$.
        Let $a\in{}\mathrm{Dom}(C_A)$ be an agent in this execution.
        Let $a'\notin \mathrm{Dom}(C_A)$ be an agent,
        and $C'=C\cup\{a'\mapsto C_A(a)\}$.
A set $\mathfrak{E}_a$ of executions
starting in configuration $C'$
is a \emph{shadow extension} of the execution~$E$
around the agent~$a$ if the following conditions hold:
\begin{itemize}
        \item removing $a'$ from each configuration
                in any execution in $\mathfrak{E}_a$ yields $E$;
        \item for each moment during the execution, 
                there is
                a corresponding execution in $\mathfrak{E}_a$
                such that $a$ and $a'$ have the same state
                at that moment.
\end{itemize}
        The added agent $a'$ is a \emph{shadow agent},
        and elements of $\mathfrak{E}_a$ are
        \emph{shadow executions}.
A protocol~$P$ is \emph{shadow-permitting}
if for every configuration~$C$
there is a fair execution starting from $C$ that
        has a~shadow extension around each agent~$a\in{}\mathrm{Dom}(C_A)$.
\end{definition}

Note that the executions in $\mathfrak{E}_a$ might not be fair
even if $E$ is fair.

Not all  population protocols
are shadow-permitting.
For example, consider a protocol with one input state $q_0$,
additional states $q_+$ and $q_-$,
and
one transition $(q_0,q_0)\to(q_+,q_-)$.
As the number of agents in
the states $q_+$ and $q_-$ is always the same, 
one can't
add a single extra agent going from state
$q_0$ to state $q_+$.

\begin{lemma}
        \label{shadow-unreliable}
        All protocols with unreliable communication are shadow-permitting.
\end{lemma}

The intuition behind the proof is the following.
We construct a fair execution together with the shadow executions
and keep track what states can be reached by the shadow agents.
The set of reachable states will not shrink, as the shadow agent
can always just fail to update.
If an agent $a$ tries to move from a state $q$ to a state $q'$
not reachable by the corresponding
shadow agent in any of the shadow executions,
we ``split'' the shadow execution reaching $q$:
one copy just stays in place, and in the other the shadow agent $a'$
takes the place of $a$ in the transition while $a$ keeps the old state.
In the main execution there is no $a'$ so $a$ participates in the
interaction but fails to update.
Afterwards we restart the process of building a fair execution.

\begin{proof}%[Proof of Lemma~\ref{shadow-unreliable}]
        We construct an execution and the families $\mathfrak{E}_a$
        in parallel, then show that the resulting execution $E$ is fair.
        We say that  a state $q$ is \emph{$a$-reachable}
        after $k$ transitions,
        if there is an execution in $\mathfrak{E}_a$
        such that $a'$ has state $q$ after $k$ transitions.
        The goal of the construction is to ensure
        that
        the set of $a$-reachable states grows as $k$ increases
        and contains the state of $a$ after $k$ transitions.

        Consider an initial configuration $C$.
        We build the execution $E$ and its
        shadow extensions $\mathfrak{E}_a$
        for each $a\in\mathrm{Dom}(C_A)$
        step by step.
        Initially, $E=(C)$ and
        $\mathfrak{E}_a$ has exactly one execution,
        namely $(C\cup\{a'\mapsto{}C_A(a)\})$.
        We pick an arbitrary fair continuation $E^\infty$
        starting with $E$.

        At each step we extend $E=(E_0=C,E_1,\ldots,E_k)$
        by one configuration and
        update $\mathfrak{E}_a$ for each $a\in\mathrm{Dom}(C_A)$.
        Consider the next configuration in $E^\infty$,
        which we can denote $E^\infty_{k+1}$.
        By definition there exists a set of agents $A^\odot$
        such that
        $E^\infty_k\xrightarrow{A^\odot}E^\infty_{k+1}$.
        We consider the following cases.

        \textit{Case 1:}
        For each agent $a$
        the state $E^\infty_{k+1}(a)$
        is $a$-reachable (after $k$ transitions).

        We set $E_{k+1}=E^\infty_{k+1}(a)$
        and keep the same $E^\infty$.
        In other words, we just copy the next transition from $E^\infty$.
        Then for each
        agent $a\in\mathrm{Dom}(C_A)$
        and for each
        $E'_a\in\mathfrak{E}_a$
        we set
        $(E'_a)_{k+1}=E_{k+1}\cup\{a'\mapsto(E'_a)_k(a')\}$,
        i.e. say that $a'$ fails to update its state.

        \textit{Case 2:}
        For each active agent $a^\odot\in{}A^\odot$
        the state $E^\infty_{k+1}(a^\odot)$
        is $a^\odot$-reachable,
        but there is a passive agent $a\notin{}A^\odot$
        such that
        the state $E^\infty_{k+1}(a)$
        is not $a$-reachable (after $k$ transitions).

        We construct $E_{k+1}$ such that
        $E_{k+1}(a^\odot)=E^\infty_{k+1}(a^\odot)$
        for each active $a^\odot\in{}A^\odot$,
        and
        $E_{k+1}(a)=E_{k}(a)$ for each passive agent
        $a\in\mathrm{Dom}(C_A)\setminus{}A^\odot$.
        In other words, all the active agents perform the update,
        but all the passive agents fail to update.
        The message packets are still consumed or created
        as if we performed the transition
        $E_k=E^\infty_k\xrightarrow{A^\odot}E^\infty_{k+1}$,
        i.e. $(E_{k+1})_P = (E^\infty_{k+1})_P$.
        As $E^\infty$ is now not a continuation of $E$,
        we replace $E^\infty$ with an arbitrary
        fair continuation of our new $E$.
        Then for each
        $E'_a\in\mathfrak{E}_a$
        we set
        $(E'_a)_{k+1}=E_{k+1}\cup\{a'\mapsto(E'_a)_k(a')\}$
        like in the previous case.
        Also, for each
        passive agent $a\in\mathrm{Dom}(C_A)\setminus{}A^\odot$
        we add a trajectory
        $E''_a$ to $\mathfrak{E}_a$
        obtained by modifying an existing trajectory
        $E'_a\in\mathfrak{E}_a$
        such that $(E'_a)_k(a')=(E'_a)_k(a)$.
        We set $(E''_a)_{k+1}(a') = E^\infty_{k+1}(a)$,
        and keep everything else the same as in $E'_a$.
        In other words, we make $a'$ perform the update that
        $a$ would perform in $E^\infty$.

        \textit{Case 3:}
        There is an active agent $a\in{}A^\odot$ such that
        the state $E^\infty_{k+1}(a)$
        is not $a$-reachable (after $k$ transitions).

        We set $(E_{k+1})_A=(E_{k})_A$,
        i.e. we say that all the agents fail to update.
        The message packets are still consumed or created
        as if we performed the transition
        $E_k=E^\infty_k\xrightarrow{A^\odot}E^\infty_{k+1}$,
        i.e. $(E_{k+1})_P = (E^\infty_{k+1})_P$.
        As $E^\infty$ is now not a continuation of $E$,
        we replace $E^\infty$ with an arbitrary
        fair continuation of our new $E$.
        Then for each
        $E'_a\in\mathfrak{E}_a$
        we set
        $(E'_a)_{k+1}=E_{k+1}\cup\{a'\mapsto(E'_a)_k(a')\}$
        (like in the previous two cases).
        Also, for each
        active agent $a\in{}A^\odot$
        we add a trajectory
        $E''_a$ to $\mathfrak{E}_a$
        obtained by modifying an existing trajectory
        $E'_a\in\mathfrak{E}_a$
        such that $(E'_a)_k(a')=(E'_a)_k(a)$.
        We set $(E''_a)_{k+1}(a') = E^\infty_{k+1}(a)$,
        and keep everything else the same as in $E'_a$.
        In other words, we allow $a'$ to update its
        state in the way $a$ would do in $E^\infty$.

We now prove that the above construction is always correctly defined
and yields a fair execution $E$ together with shadow extensions
around each agent.

First we show that we always continue $E$ in a valid way, i.e.
$E_k\xrightarrow{A^\odot}E_{k+1}$.
In the first case it is true by construction
as $E_k=E^\infty_k$ and $E_{k+1}=E^\infty_{k+1}$.
In the second and the third case, we modify the states of some agents
in the second configuration of a valid transition
$E^\infty_k\xrightarrow{A^\odot}E^\infty_{k+1}$
by assigning them the states from the first configuration.
Such changes clearly cannot violate
population preservation and message preservation.
State preservation is satisfied because we replace the agent's
state in the second configuration with the state from the first configuration.
The case split between the cases~2 and~3 ensures reliance on active agents;
we either make sure that all the active agents update their state,
or none of them.
Therefore, all the conditions of the Definition~\ref{def:unreliable-protocol}
are satisfied and the changed transition is also present
in the protocol with unreliable communication.

As the updated execution $E$ is a valid finite execution,
we can find a fair continuation $E^\infty$ as the fairness condition is eventual.

When we extend the executions
in the shadow extensions
by repeating the same state,
 we just use
possibility to add passive agents to add $a'$ to the valid transition from $E$,
then observe that making a passive agent fail to update
is always allowed in an  protocol with unreliable communication.

When we add new trajectories in cases~2 and~3,
we use possibility to add passive agents to add $a'$ to the valid transition from~$E$,
then we use agent anonymity to swap the state changes of $a$ and $a'$,
then we use unreliability to make the (passive) agent $a$ fail to update the state,
as well as either all the passive or all the agents from $\mathrm{Dom}(C_A)$.

So far we know that the construction can be performed and yields
a valid execution $E$
and some valid executions in each $\mathfrak{E}_a$.
Now we check that each $\mathfrak{E}_a$ is a shadow extension around $a$,
and $E$ is fair.
We observe that our construction indeed only increases the set of $a$-reachable
states as the number of transitions grows.
Furthermore, at each step either agent $a$ moves to an $a$-reachable state,
or $a$ stays in an $a$-reachable state,
thus $\mathfrak{E}_a$ is indeed a shadow extension around the agent $a$.
Whenever the fair continuation $E^\infty$ is changed,
for at least one agent $a$ the set of $a$-reachable states
strictly increases.
As the set of agents is finite and cannot change by agent conservation,
and the set of states is finite,
all but a finite number of steps correspond to the case~1.
Therefore from some point on $E^\infty$ does not change
and $E$ coincides with it, and therefore $E$ is fair.

        This concludes the proof of the lemma.
\end{proof}

We also
use a straightforward generalisation of
the truncation lemma from \cite{journals/dc/AngluinAER07}.
The lemma says that all large \emph{amounts} of agents are equivalent
for the notion of stable consensus.

\begin{definition}
        A protocol is \emph{truncatable} if
there exists a number $K$ such that
for every stable consensus
        adding an extra agent with a state $q$
that is already represented by at least $K$ other agents
        yields a stable consensus.
\end{definition}

%The following two lemmas follow the standard copycat arguments.

\vskip-3.5mm

\begin{lemma}[adapted from \cite{journals/dc/AngluinAER07}]
        \label{lemma:truncatable}
        All protocols (not necessarily with unreliable communication) are truncatable.
\end{lemma}

\begin{proof}
Every configuration can be summarised
        by an element of $\mathbb{N}^{Q\cup{}M}$
(each state is mapped to the number of agents in this state,
each message is mapped to the number of packets with this message).
In other words, we can forget the  identities and consider
the multiset of states and messages.
If a configuration is a consensus (correspondingly, stable consensus),
all the configurations with the same multiset of states
        and messages are also
        consensus configurations (correspondingly, stable consensus
        configurations).
        The set $\overline{ST}$ of elements of $\mathbb{N}^{Q\cup{}M}$
\emph{not} representing stable consensus configurations
        is upwards closed, because reaching a state with
a different local output value cannot be impeded by adding agents
or packets.
Indeed, if we can reach a configuration $C_{\overline{ST}}$
with some state $q$ present,
we can always use addition of passive agents to each transition of the path
and still have a path of valid transitions from a larger configuration to
some configuration $C^*_{\overline{ST}}$ with state $q$ still present.
By possibility to ignore extra packets, we can also allow
additional packets in the initial configuration.
By Dickson's lemma, the set $\overline{ST}$
of non-stable-consensus state multisets
has a finite set of minimal elements $\overline{ST}_{min}$.
We can take $K$ larger than all coordinates of all minimal elements.
Then adding more agents with the state that already has
at least $K$ agents
leads to increasing a component larger than $K$ in the multiset of states.
This cannot change any component-wise comparisons
with multisets from $\overline{ST}_{min}$,
and therefore belonging to $\overline{ST}$
and being or not a stable consensus.
\end{proof}

\begin{remark}
A specific bound on the truncation threshold $K$ can be obtained
using the Rackoff's bound for the size of configuration necessary for covering
in general Vector Addition Systems \cite{journals/tcs/Rackoff78}.
\end{remark}

\begin{lemma}
If a predicate $\varphi$ can be
implemented by a shadow-permitting truncatable protocol,
then $\varphi$ is a counting predicate.
\end{lemma}

\begin{proof}
Let $K$ be the truncation constant. We claim that $\varphi$
can be expressed as a combination of threshold predicates
with thresholds no larger than $|Q|\times K$.

More specifically, we prove an equivalent statement:
adding $1$ to an argument already larger than $|Q|\times K$ doesn't change the output
value of $\varphi$.
Let us call the state corresponding to this argument $q$.
Indeed, consider any corresponding input configuration.
We can build a fair execution starting in it with shadow extensions
around each agent.
As the predicate is correctly implemented, this fair execution has to
reach a stable consensus.
By assumption (and pigeonhole principle),
more than $K$ agents from the state $q$
end up in the same state.
By definition of shadow extension, there is an execution
starting with one more agent in the state $q$,
and reaching the same stable consensus but
with one more agent in a state with more than $K$ other ones
(which doesn't break the stable consensus).
Continuing this finite execution to a fair execution
we see that the value of $\varphi$ must be the same.
This concludes the proof.
\end{proof}

For the lower bound, we adapt the following lemma from
\cite{journals/dc/AngluinAER07}.

\begin{lemma}
        \label{lm:unreliableContainCounting}
        All counting predicates can be implemented
        by immediate observation protocols
        (possibly with unreliable communication),
        even if the fairness condition is replaced with an arbitrary
        different (activity-ensuring) one.
\end{lemma}

\begin{proof}
We have already observed that immediate observation population protocols
do not change if we add unreliability.
It was shown in \cite{journals/dc/AngluinAER07} that
immediate observation population protocols
implement all counting predicates.
        Moreover, the protocol $(k,k)\mapsto(k+1,k); (k,n)\mapsto(n,n)$
        provided there for threshold predicates
has the state of each agent increase monotonically. It is easy to see
that ensuring activity is enough for this protocol to converge to a state
where no more configuration-changing transitions can be taken.
Also, the construction for boolean combination of predicates
via direct product of protocols
used in \cite{journals/dc/AngluinAER07}
converges as long as the protocols for the two arguments converge.
Therefore it doesn't need any extra restrictions on the fairness condition.
\end{proof}

Theorem~\ref{thm:implementable}
now follows from the fact that all the  protocols with unreliable communication
are shadow-permitting (by Lemma~\ref{shadow-unreliable}) and
truncatable (by Lemma~\ref{lemma:truncatable}), therefore
they only implement counting predicates.
By Lemma~\ref{lm:unreliableContainCounting} all counting predicates
can be implemented.

\section{Non-monotonic impact of unreliability}
\label{sec:message-based}

In this section we observe that, surprisingly,
while delayed transmission protocols
and queued transmission protocols
are more powerful than immediate observation population protocols,
their unreliable versions are strictly less expressive
than  immediate observation population protocols (possibly with unreliable communication).
%Moreover, the expressive power becomes very low.
%For simplicity, we only prove the following restriction on the
%expressive power.
%Let a  well-specified protocol with unreliable communication
%have a single input state,
%and let each transition depends on the state of only a single agent
%(but possibly also some packets).
%Then the protocol cannot distinguish 
%(i.e. must provide the same result)
%the initial configuration with one agent
%from the initial configuration with two agents.

\begin{definition}
A protocol is \emph{fully asynchronous} if
        for each allowed transition $(C,A^\odot,C')$ the following conditions hold.
\begin{itemize}
        \item There is exactly one active agent, i.e. $|A^\odot|=1$.
        \item No passive agents change their states.
        \item Either the packets are only sent or the packets are only consumed,
                i.e. $\mathrm{Dom}( C_P)\subseteq\mathrm{Dom}( C'_P)$ or $\mathrm{Dom}( C_P)\supseteq \mathrm{Dom}( C'_P)$.
                Packet contents do not change,
                i.e.
                $C_P\mid_{\mathrm{Dom}(C_P)\cap\mathrm{Dom}(C'_P)} =
                C'_P\mid_{\mathrm{Dom}(C_P)\cap\mathrm{Dom}(C'_P)}$.
\end{itemize}
\end{definition}

It turns out that given unreliable communication such protocols
can check presence of states but cannot count.
As our old notion of ensuring activity doesn't force
any messages to be ever received, we need a slightly stronger fairness 
condition for any positive claims.

\begin{definition}
A fairness condition \emph{ensures communication}
if the following two conditions hold in every fair run.
\begin{enumerate}
\item
If the agent states $C_A$ do not change after some moment,
from each configuration occurring after some later moment 
there is no possible transition changing $C_A$.
\item
If the set of messages present in $C_P$ (ignoring multiplicities)
does not change after some moment, 
then for each configuration after some later moment
there is no possible transition that creates
a packet with a new message.
\end{enumerate}
\end{definition}

\begin{theorem}\label{thm:async-weak}
Fully asynchronous protocols with unreliable communication
compute exactly the predicates that are boolean combinations 
of positivity of single coordinates.

The upper bound holds under any eventual fairness condition,
while the lower bound requires a fairness conditions
that ensures communication.
\end{theorem}

The core idea of the proof is to
ensure that in a reachable situation
rare messages do not exist and cannot be created.
In other words, if there is a packet with some message,
or if such a packet can be created,
then there are many packets with the same message.
This makes irrelevant both the production of new messages
by agents, and the exact number of agents needing to follow
a particular sequence of transitions.
This idea has some similarity with the message saturation
construction from \cite{journals/corr/abs-1912-06578},
but here the production
of new messages might require consuming some of the old ones.
We choose the threshold for ``many'' packets
depending on the number of messages that do \emph{not} yet
have ``many'' packets.
The threshold ensures that a new message will become abundant
before we exhaust the packets for any previously numerous message.

\begin{definition}
        The \emph{in-degree} of a fully asynchronous protocol
        is the maximum number of messages consumed in a single transition.

        The \emph{supply} of a message $m\in M$ in configuration $C$
        is the number of packets in $C$ with the message $m$,
        i.e. $|C_P^{-1}(m)|$.

        Let $F(x,y,z,n,k)=(32(xyzn+1))^{32(xyzn+1)-2k}$.
        An \emph{abundance set} is the largest set $M^\infty\subseteq M$
        such that the supply of each message in $M^\infty$
        is at least $F(|Q|,|M|, d,|C_A|,|M^\infty|)$
        where $d$ is the in-degree.
        As $F$ decreases in the last argument,
        the abundance set $M^\infty(C)$ is well-defined.
        A message $m$ is \emph{abundant} in configuration $C$
        if it is in the abundance set, i.e. $m\in M^\infty(C)$.
        A message $m$ is \emph{expendable} at some moment
        in execution $E$ if it is abundant in some configuration
        that has occurred in $E$ before that moment.
        A packet is \emph{expendable} if it bears an expendable
        message.

        An execution $E$ is \emph{careful}
        if no transition that decreases the supply of non-expendable messages
        changes agent states.
\end{definition}

\begin{remark}
The function $F$ is chosen to make its rate of growth obviously sufficient
in the following calculations. 
A much smaller function would suffice
for a more tedious analysis.
\end{remark}

\begin{lemma}\label{LemmaCarefulExecutions}
Every  fully asynchronous protocol with unreliable communication
has a careful fair execution starting from any configuration
without message packets.

Moreover, if the protocol is well-specified,
there is a careful fair execution that runs each packet-consuming 
transition twice in a row, failing to update the state the first time,
until  stable consensus is reached.
\end{lemma}

\begin{proof}
We start with an execution with only the initial configuration.

In the first phase,
as long as it is possible to create a packet with a non-expendable message
(without making the execution careless),
we do it while consuming the minimal possible number 
of packets with expendable messages.
After creating each packet we increase the abundance set if possible.

In the second phase,
a long as it is possible to consume a packet with a non-expendable message,
we do it (but fail to update the agent states).

In the third phase we reach a stable consensus by consuming the minimal
number of packets.
We call the end of the third phase the target moment.
Afterwards we pick an arbitrary fair continuation.

We now prove that each abundance set
with a new message
obtained during the first phase includes
all the previous abundance sets.
We only use the ways to
create a new non-expendable packet
that do not require consuming any non-expendable packets.
Indeed, consuming a non-expendable packet
is not allowed to change the internal state
by definition of carefulness,
and cannot create any new messages by definition
of a fully asynchronous protocol.
Note that reaching the internal state
that can create a new non-expendable packet
can take most $|Q|\times n$ transitions
as all the expendable packets are already available for consumption
and thus there is no reason to repeat the same internal state 
of the same agent twice.
Therefore creating an additional non-expendable packet
can consume at most $|Q|\times n\times d$ packets.
To make the supply of some message reach
$F(|Q|,|M|,d,n,k+1)$,
we need to
repeat this at most $F(|Q|,|M|,d,n,k+1)\times |M|$ times
consuming at most
$F(|Q|,|M|,d,n,k+1)\times M|\times |Q|\times n\times d$ expendable packets.
We might consume twice as many expendable packets if we want to 
fail every other packet consumption transition.
As $3\times F(|Q|,|M|,n,d,k+1)\times |M|\times |Q|\times n \times d<F(|Q|,|M|,d,n,k)$,
all the expendable messages together with this message
form an abundance set.

In the second phase,
we run consumption in at most $|Q|$ states;
reaching each of them requires at most $|Q|$ transitions.
Thus the state changes
consume at most $|Q|^2\times d$ expendable packets.
Note that consuming a non-expendable packet requires consuming
at most $d$ expendable packets.
As the supply of each non-expendable message is less than
$F(|Q|,|M|,d,n,|M^\infty|+1)$,
we consume at most $d\times(|Q|^2+|M|\times F(|Q|,|M|,d,n,|M^\infty|+1))$.
We also could have spent twice as many expendable messages if the 
non-expendable messages were not the limiting factor.
Therefore we still have more than 
$F(|Q|,|M|,d,n,|M^\infty|+1)>4\times |Q|\times n\times d$
packets with each expendable message left
by the time there are
no non-expendable packets
that can be received in a reachable state
and no possibility to create a non-expendable packet.

A reachable stable consensus exists
if the protocol computes some predicate.
As it is impossible to produce or consume new non-expendable messages,
we cannot violate the carefulness property.
Moreover, we can reach it while spending at most
$|Q|\times n\times d$ expendable packets 
(or twice as many if we fail to update the state every second time).
That many packets are available,
so producing new expendable packets is not required.

We see that the construction indeed provides a careful fair execution.
Thus the lemma is proven.
\end{proof}

As the execution obtained via the previous lemma 
wastes a lot of messages, we can add one more agent 
to make use of those messages.

\begin{lemma}
\label{lm:message-based-copycat}
Consider a fully asynchronous protocol with unreliable communication
that computes some predicate.

Then for any input configuration, adding one more agent in an already 
present
input state
cannot change the value of the predicate.
\end{lemma}

\begin{proof}
Consider a careful execution constructed by Lemma~\ref{LemmaCarefulExecutions}.
Consider an extra agent that we want to use as a copycat of an existing one,
which we call target.

If a transition performed by the target agent sends messages, so does
the copycat agent. 
If a transition requires receiving messages and the target agent updates
the state,
we cancel the previous transition where the target agent failed to update the state
after consuming the same messages,
and let the copycat agent receive those messages and update the state.
Thus the copycat agent always mimics the state of the target agent.

Additionally, we extend phase two of the execution to consume the non-expendable
messages sent by the copycat agent.
They are the same as the target agent has sent, 
and
there is a reserve of expendable messages for 
consuming these non-expendable messages (those that can be consumed
in some reachable state).

As consumption of expendable messages did not allow to emit
any non-expendable messages after reaching the stable consensus, the same must 
be true when we add the copycat agent as the set of reachable agent states 
without producing or consuming non-expendable messages
is the same.
But then the set of all the reachable states is the same,
and we get a stable consensus with the same answer.

As the protocol is well-specified, this concludes the proof.
\end{proof}

\begin{corollary}
\label{lm:async-only-positivity}
A predicate computed by a fully asynchronous protocol
with unreliable communication
only depends on which coordinates are positive.
\end{corollary}

\begin{proof}
Consider two configurations with the same set of represented 
input states.
By repeated addition of copycat agents we can prove that the 
predicate value for either of configurations is the same as the 
predicate value for their union.
\end{proof}

It is clear that the predicates that only depend on the set of
positive coordinates can be computed.

\begin{lemma}
\label{lm:async-strong}
For any fairness condition ensuring communication, and for any 
predicate only depending on positivity of arguments,
there is a fully asynchronous protocol computing that predicate.
\end{lemma}

\begin{proof}
We just describe the protocol informally.
The messages correspond to the input states.
The states correspond to nonempty states of the input state
(which are known to the agent to be initially present).
An agent can send a message corresponding to an initial state 
in the agent's set.
An agent can receive a message and add the corresponding initial state
to the set.
An agent has output value equal to the value of the predicate on the 
input where all the input states from the agent's set get the value $1$,
while the others get $0$.

Ensuring communication implies that the only stable situation is when
all the initially present input states are reflected in message packets,
and are also reflected in the sets of all the agents.
\end{proof}

The theorem now follows from Corollary~\ref{lm:async-only-positivity}
and Lemma~\ref{lm:async-strong}.

\begin{remark}
This result doesn't mean that fundamentally asynchronous nature
of communication prevents us from using any expressive models for
verification of unreliable systems.
It is usually possible to keep enough state to
implement, for example,
immediate observation via request and response.
\end{remark}

\section{Conclusion and future directions}
\label{sec:conclusion}

We have studied unreliability based on message loss,
a practically motivated approach
to fault tolerance in population protocols.
We have shown that inside a general framework
of defining  protocols with unreliable communication we can prove
a specific structural property
that bounds
the expressive power of  protocols with unreliable communication
by the
expressive power of
immediate observation population protocols.
Immediate observation population protocols
permit verification of many useful properties,
up to well-specification, correctness and
reachability between counting sets,
in polynomial space.
We think that relatively low complexity of verification
together with inherent unreliability tolerance
and locally optimal expressive power under atomicity violations
motivate further study and use of such protocols.

It is also interesting to explore if
for any class of protocols adding unreliability makes
some of the verification tasks easier.
Both
complexity and expressive power
implications
of unreliability
can be studied
for models with larger per-agent memory,
such as community protocols, PALOMA and
mediated population protocols.
We also believe that some models even more restricted than
community protocols but still permitting a multi-interaction
conversation are an interesting object of study both
in the reliable and unreliable settings.

\subsection*{Acknowledgements}
I~thank Javier Esparza
for useful discussions and
the feedback on the drafts of the present article.
I~thank Chana Weil-Kennedy for useful discussions.

This work is an extended version of \cite{DBLP:conf/algosensors/Raskin21},
differing in the inclusion of full proofs as well as more precise characterisation
of expressive power of fully asynchronous protocols.
I~thank the anonymous reviewers both of the previous and of the current
version
for their valuable feedback on presentation.

\bibliographystyle{plain}
\bibliography{unreliable-population-protocols}

\end{document}